\documentclass{article}
\usepackage[margin=1in]{geometry}
\usepackage{setspace}
\usepackage{graphicx} 
\usepackage{subcaption}
\usepackage{amsmath,amsthm,amssymb}
\usepackage{bm}
\usepackage{color,hyperref}

\newtheorem{theorem}{Theorem}
\newtheorem{lemma}[theorem]{Lemma}

\newtheorem{corollary}[theorem]{Corollary}
\newtheorem{definition}[theorem]{Definition}
\newtheorem{remark}{Remark}

\newcommand{\R}{\mathbb{R}}
\newcommand{\W}{\mathcal{W}}
\newcommand{\x}{\mathbf{x}}
\newcommand{\w}{\mathbf{w}}
\renewcommand{\v}{\mathbf{v}}

\newcommand{\s}{\mathbf{s}}
\renewcommand{\c}{\mathbf{c}}
\renewcommand{\u}{\mathbf{u}}

\newcommand{\z}{\mathbf{z}}

\renewcommand{\vec}[1]{{\bf #1}}
\newcommand{\vecrho}{{\bm\rho}}
\newcommand{\vecmu}{{\bm\mu}}

\newcommand{\cl}{\operatorname{cl}}
\newcommand{\diag}{\operatorname{diag}}

\DeclareMathOperator*{\argmin}{arg\,min}

\title{Mandate without Managers: Automated Market Makers as Verifiable Portfolio Products}
\author{Zachary Feinstein\thanks{Stevens Institute of Technology, School of Business, Hoboken, NJ 07030. {\tt zfeinste@stevens.edu}.} 
\and 
Ionut Florescu\thanks{Stevens Institute of Technology, School of Business, Hoboken, NJ 07030.} 
\and 
Sean O'Leary\thanks{Stevens Institute of Technology, School of Business, Hoboken, NJ 07030.}}
\date{\today}

\begin{document}

\maketitle
\begin{abstract}
Automated market makers (AMMs) are typically interpreted and evaluated as decentralized exchanges. Herein, we take the perspective envisioned by \emph{Balancer} that an AMM can also be viewed as a portfolio technology that programmatically enforces an economic mandate. In particular, we follow the geometric mean market maker (G3M) invariant employed by that protocol in order to enforce a target-weighted portfolio. We introduce a multi-asset fee structure to the G3M under which competitive arbitrage implements a band-rebalancing strategy with mis-weighting bounded ex ante, allowing compliance with the mandate to be verified directly from the pool's observable holdings. We then compare simulated G3M portfolios against the realized performance of VBIAX, EQL, and EDOW on annualized returns and tracking error against the portfolio mandate. Across these historical case studies, and using arbitrage-only order flow, the G3M is found to outperform the incumbent funds in both metrics for certain fee ranges.
\end{abstract}

\section{Introduction}\label{sec:intro}

Target-weight portfolio mandates are among the most common products in traditional asset management. Financial advisors routinely recommend allocations such as 60\% equity and 40\% fixed income. Certain such positions are so common that they are offered as mutual funds and exchange-traded funds (ETFs). For these funds, the product is \emph{both} the asset holdings \emph{and} the value of the commitment to maintain the stated allocation over time. Unlike the better known market cap weighted ETFs, these target weighted ETFs require active rebalancing which can be costly to achieve~\cite{el2023rebalancing}. 

By following a mathematical invariant, an automated market maker (AMM) offers a programmatic way to implement such a commitment. Though an AMM is typically interpreted as a rule governing the terms at which assets can be exchanged~\cite{angeris2020improved,angeris2023geometry}, the same automation also determines the composition of the reserves held on behalf of liquidity providers (LPs). In this sense, an AMM is simultaneously an exchange and a portfolio rule. From the portfolio perspective, the role of the arbitrageur changes to that of an external rebalancing agent who is compensated by the arbitrage profits extracted from the LPs. The AMM pool, therefore, maintains its mandate algorithmically and without managerial discretion.
The portfolio interpretation of an AMM is particularly direct for the geometric mean market maker (G3M) as this AMM encodes a target-weighted portfolio~\cite{evans2021liquidity,evans2021optimal}. This viewpoint has been commercialized by \emph{Balancer}~\cite{martinelli2019non} and \emph{QuantAMM}~\cite{willetts2024optimal,willetts2026pools} for crypto-assets.

We explore the G3M as a portfolio rebalancing technology both theoretically and empirically.
Section~\ref{sec:g3m} provides background on the G3M and introduces a novel fee structure under which competitive arbitrage constrains the realized portfolio weights to an explicitly characterized neighborhood of the target weights. This permits the fee-bearing G3M to be interpreted as a band-rebalanced target-weight strategy. This extends the two-asset setting of~\cite{evans2021optimal} (under a modified fee structure) to provide a closed-form characterization of the fee-induced no-arbitrage region in \emph{weight space} for a multi-asset G3M.
Section~\ref{sec:cs} compares simulated G3M portfolios with real funds to empirically assess whether AMMs can serve as a practical portfolio technology when competing against incumbent funds. This builds upon the work of~\cite{willetts2024rebalancing} to consider more realistic benchmarks against which to evaluate AMM performance. As far as we are aware, we provide the first direct comparison between an AMM-held portfolio and traditional mutual funds and ETFs. 
Across three case studies, we find that the G3M, in certain parameter ranges under purely adversarial order flow, can dominate the incumbent fund in both returns and tracking error against the economic mandate.

\section{Geometric Mean Market Maker}\label{sec:g3m}
Within this section we, first, provide a background on the G3M without fees (Section~\ref{sec:g3m-bg}). In doing so, we highlight the viewpoint that an AMM pool is, fundamentally, a portfolio rebalancing strategy. In the case of the G3M, this viewpoint was presented and supported in~\cite{evans2021liquidity} from which we draw much of this background section. Second, in Section~\ref{sec:g3m-fees}, we introduce a fee-earning variant of the G3M which satisfies the path-independence properties of~\cite{bichuch2026axioms}. Under this design, we study the arbitrageurs' optimization problem in order to characterize the no-arbitrage region. In particular, the arbitrage exists if and only if the portfolio is sufficiently mis-weighted in at least one component.

\subsection{Background: Fee-Free Setting}\label{sec:g3m-bg}
Following the discussion in, e.g.,~\cite{evans2021liquidity}, consider a liquidity pool consisting of $N$ assets with reserves $\x := (x_1,...,x_N)^\top \in \R^N_{++}$ and weights $\w := (w_1,...,w_N)^\top \in \W := \{\w \in \R^N_{++} \; | \; \sum_{i = 1}^N w_i = 1\}$. A Geometric Mean Market Maker (G3M) is an AMM defined by the invariant
\begin{equation}\label{eq:g3m}
L = f_\w(\x) := \prod_{i = 1}^N x_i^{w_i}
\end{equation}
where $L > 0$ denotes the number of LP shares. Any swap $\Delta\x \in \prod_{i = 1}^N (-x_i,\infty)$ at this G3M is done in such a way that $f_\w(\x+\Delta\x) = L$; it is for this reason that~\eqref{eq:g3m} is referred to as the invariant. Importantly, an LP share entitles the holder to the proportionate share of the underlying reserves $\x$ via an invariant level update ($L' = (1+\alpha)L = f_\w((1+\alpha)\x)$ with proportional deposit or withdrawal $\alpha > -1$). For this reason we will often define $\x/L$ as the normalized portfolio held by the AMM pool for LPs.

Because the invariant~\eqref{eq:g3m} defines how swaps can be executed, the marginal price $p_{ij} > 0$ of any asset $i$ in terms of asset $j$ in the pool can be found via implicit differentiation~\cite{angeris2020improved}. Specifically, given the reserves $\x$, this marginal price is given by
\begin{equation}\label{eq:g3m-price}
p_{ij} := \frac{w_i x_j}{w_j x_i}.
\end{equation}
If we impose external (dollar-denominated) prices $\s := (s_1,...,s_N)^\top \in \R^N_{++}$ on these assets, then the arbitrageurs would seek to maximize value extracted from the pool, i.e., 
\[\max_{\Delta\x \in \prod_{i = 1}^N (-x_i,\infty)} \left\{-\s^\top \Delta\x \; | \; f_\w(\x+\Delta\x) = L\right\}.\]
Let the resulting (normalized) reserves be denoted by $\x^*(\s) := (\x+\Delta\x^*)/L$. 
\begin{lemma}\label{lemma:arb-0fee}
For any $\s \in \R^N_{++}$, the realized (normalized) portfolio holdings are provided by $\x_i^*(\s) = \frac{w_i}{s_i} V(\s)$ for any asset $i$ where $V(\s) := \prod_{i = 1}^N (s_i/w_i)^{w_i}$ is the normalized portfolio value (i.e., $V(\s) = \s^\top \x^*(\s)$).
\end{lemma}
\begin{proof}
This follows directly from Equations (5) and (6) of~\cite{evans2021liquidity}.
\end{proof}

The key property of the G3M that we will take advantage of is that the weight of asset $i$ ($s_i\x_i^*(\s)/V(\s)$) in the portfolio $\x^*(\s)$ held by LPs is provided by $w_i$.
\begin{corollary}\label{cor:tracking}
For any $\s \in \R^N_{++}$, the realized portfolio weight of asset $i$ is equal to the G3M's weight. That is,
\[\frac{s_i x_i^*(\s)}{\s^\top \x^*(\s)} = w_i \quad \forall i.\]
\end{corollary}
\begin{proof}
This perfect weight matching follows directly from the construction of $\x^*(\s)$ and $V(\s)$ in Lemma~\ref{lemma:arb-0fee}.
\end{proof}

\subsection{Bounded Mis-Weighting with Fees}\label{sec:g3m-fees}

The frictionless G3M framework described in Section~\ref{sec:g3m-bg} treats all feasible reserve updates through the same invariant $L = f_\w(\x)$. Once fees are introduced, however, different types of operations should be distinguished; in particular, proportional liquidity provision and redemption should remain frictionless, with fees only applied to the non-representative basket liquidity operations or direct swapping. In contrast to prior works (e.g.,~\cite{angeris2020improved,angeris2023geometry,bichuch2026axioms}), herein we introduce a single invariant system for both liquidity provision and swapping. This is accomplished by breaking any market operation into two pieces: (1) a shared proportional mint (or burn) component that scales all reserves uniformly and remains fee-free; and (2) the residual non-proportional adjustment that changes the portfolio composition and to which fees are applied.
Furthermore, following the discussion in~\cite{bichuch2026axioms}, we define a novel fee structure that prevents the strategic avoidance of fees while remaining path independent for pure trade splitting. We note that this construction is distinct from that utilized in prior works (e.g.,~\cite{evans2021optimal}).

Let $\x' := \x + \Delta\x$ denote the post-update reserve vector for initial reserves $\x \in \R^N_{++}$ and $\Delta\x \in \prod_{i = 1}^N (-x_i,\infty)$. We can define the proportional change in reserves by $\vecrho := (\rho_1,...,\rho_N)^\top = \diag(\x)^{-1}\Delta\x$. In this way, we can equivalently define $\x' = (I + \diag(\vecrho))\x$. With this convention, we can define the shared proportional component of the reserve update by $\alpha(\vecrho) := \min_i \rho_i$. The quantity $1 + \alpha(\vecrho) > 0$ is the largest proportional scaling common to every asset, and we interpret it as the shared mint/burn component of the update; the remaining, non-proportional component is interpreted as a residual swap and is the only part of the update to which fees should apply. In this way, letting $\tilde\rho_i := \frac{\rho_i - \alpha(\vecrho)}{1+\alpha(\vecrho)} \geq 0$ be the residual relative update after the proportional component, we find $\x' = (1+\alpha(\vecrho))(I + \diag(\tilde{\vecrho}))\x$. 

\begin{definition}\label{defn:g3m}
Fix a fee parameter $\gamma \in [0,1)$. Consider initial reserves $\x \in \R^N_{++}$ and LP shares $L > 0$ for a G3M with weights $\w \in \W$. Given a market operation $\Delta\x \in \prod_{i = 1}^N (-x_i,\infty)$, the updated liquidity is defined by $L' = L g_\w^\gamma(\diag(\x)^{-1}\Delta\x)$ for
\begin{equation}\label{eq:g}
g_\w^\gamma(\vecrho) := (1 + \alpha(\vecrho))^\gamma \prod_{i = 1}^N (1 + \rho_i)^{(1-\gamma)w_i} = (1+\alpha(\vecrho)) \prod_{i = 1}^N (1 + \tilde\rho_i)^{(1-\gamma)w_i}
\end{equation}
for any $\vecrho \in (-1,\infty)^N$. The resulting change in liquidity is
\begin{equation}\label{eq:ell}
\ell_\w^\gamma(\Delta\x;\x,L) := L [g_\w^\gamma(\diag(\x)^{-1}\Delta\x) - 1].
\end{equation}
\end{definition}

The change in liquidity $\ell_\w^\gamma(\Delta\x;\x,L)$ defined in~\eqref{eq:ell} provides the number of LP tokens minted (if positive) or burned (if negative) for the investor. A pure swap is the case in which $\ell_\w^\gamma(\Delta\x;\x,L) = 0$. Under this construction, the derived bid-ask matrix~\cite{schachermayer2004fundamental} is such that $\pi_{ij}(\x) = \frac{\gamma + (1-\gamma)w_j}{(1-\gamma)w_i} \frac{x_i}{x_j}$ is the (marginal) units of asset $i$ that must be paid to receive one (marginal) unit of asset $j \neq i$ (with $\pi_{ii}(\x) \equiv 1$ for any asset $i$). Since $\pi_{ij}(\x) \pi_{jk}(\x) > \pi_{ik}(\x)$ when $\gamma > 0$, the fees generate a strictly positive bid-ask spread which cannot be reduced by routing through an intermediate asset. This bid-ask spread makes clear that the fee structure introduced in Definition~\ref{defn:g3m} is distinct from the proportional fees applied in, e.g.,~\cite{evans2021optimal}.

\begin{remark}
The construction of the G3M with fees as provided in Definition~\ref{defn:g3m} satisfies three important consistency properties:
\begin{itemize}
\item If $\gamma = 0$ then $g_\w^0(\vecrho) = \prod_{i = 1}^N (1 + \rho_i)^{w_i}$ by construction. Recalling from~\eqref{eq:g3m} that $L = f_\w(\x)$, we recover $\ell_\w^0(\Delta\x;\x,f_\w(\x)) = f_\w(\x+\Delta\x) - f_\w(\x)$.
\item Fix $\vecrho \in (-1,\infty)^N$, then $\gamma \in [0,1) \mapsto g_\w^\gamma(\vecrho)$ is non-increasing. As a consequence, the cost of a market operation in terms of LP shares ($-\ell_\w^\gamma(\Delta\x;\x,L)$) is non-decreasing in the fee parameter; this difference from the fee-free $\gamma = 0$ case is exactly the fee charged for that operation which is distributed pro-rata to the existing LPs through the increased holdings of each normalized share.
\item Consider $\vecrho = \alpha\vec{1}$ for some constant $\alpha \in (-1,\infty)$. By construction $g_\w^\gamma(\vecrho) = 1+\alpha$. Substituting into~\eqref{eq:ell} gives $\ell_\w^\gamma(\alpha\x;\x,L) = L\alpha$, i.e., the exact proportional change in liquidity with no dependence on $\gamma$. Proportional minting and burning are therefore fee-free for every $\gamma \in [0,1)$.
\end{itemize}
\end{remark}

We call this the path-independent case as $g_\w^\gamma(\vecrho) = (1 + \alpha(\vecrho))^\gamma g_\w^0(\vecrho)^{(1-\gamma)}$ is simply a geometric interpolation between the fee-free G3M liquidity update and the maximal common proportional update. Due to the strong path-independence of the fee-free G3M (see, e.g.,~\cite{bichuch2026axioms}), we recover comparable properties with the fee so long as the minimizing component for $\alpha(\vecrho)$ is consistent for any decomposition of a market operation.

\begin{lemma}\label{lemma:path-independence}
Fix $\gamma \in [0,1)$ and $\w \in \W$. Let $\vecrho \in (-1,\infty)^N$ and let $(\vecrho^{(k)})_{k = 1}^K \subseteq (-1,\infty)^N$ be any decomposition of that market operation, i.e., $1 + \rho_i = \prod_{k = 1}^K (1 + \rho_i^{(k)})$ for every asset $i$. Then 
\begin{equation}\label{eq:path-independence}
g_\w^\gamma(\vecrho) \geq \prod_{k = 1}^K g_\w^\gamma(\vecrho^{(k)}).
\end{equation}
If $\gamma = 0$ then~\eqref{eq:path-independence} holds with equality for every decomposition. If $\gamma \in (0,1)$ then equality holds in~\eqref{eq:path-independence} if and only if the sub-operations share a common minimizing asset, i.e.,
$\bigcap_{k = 1}^K \argmin_i \rho_i^{(k)} \neq \emptyset$.
\end{lemma}
\begin{proof}
See Appendix~\ref{proof:lemma:path-independence}.
\end{proof}

Notably, this interpolation between the minimal asset and the fee-free G3M allows us to utilize Corollary~\ref{cor:tracking} in order to define the mis-weighting that the fees introduce. In particular, following this construction, the realized weights (under no-arbitrage) satisfy $\hat\w \in (1-\gamma)\w + \gamma\cl\W$ and, in particular, $\hat w_i \geq (1-\gamma)w_i$ for every asset $i$. This is formally proven within the following theorem in which it is assumed that arbitrageurs act to maximize the value extracted from the pool subject to the fees charged, i.e.,
\[\max_{\Delta\x \in \prod_{i = 1}^N (-x_i,\infty)} \left\{-\s^\top\Delta\x \; | \; g_\w^\gamma(\diag(\x)^{-1}\Delta\x) = 1\right\}.\]

\begin{theorem}\label{thm:tracking-fees}
Fix the fee parameter $\gamma \in [0,1)$ and let the realized weights of an LP share be given by $\hat{w}_i := \frac{s_i x_i}{\s^\top\x}$ for every asset $i$ for external prices $\s \in \R^N_{++}$. 
No profitable arbitrage trade occurs if and only if $\hat{w}_i \geq (1-\gamma)w_i$ for every asset $i$.
\end{theorem}
\begin{proof}
See Appendix~\ref{proof:thm:tracking-fees}.
\end{proof}

\begin{corollary}\label{cor:tracking-fees}
The no-arbitrage condition of Theorem~\ref{thm:tracking-fees} imposes that, under active arbitrage, the realized weights live in the simplex $\hat\w \in (1-\gamma)\w + \gamma \cl\W$. In particular, this implies that asset $i$ lives within the mis-weighting band $\hat{w}_i \in [(1-\gamma)w_i \, , \, \gamma + (1-\gamma)w_i]$.
\end{corollary}
\begin{proof}
See Appendix~\ref{proof:cor:tracking-fees}.
\end{proof}

\begin{remark}
Corollary~\ref{cor:tracking-fees} bounds the allocation drift $\sum_{i = 1}^N |\hat{w}_i - w_i| \leq \gamma \sup_{\v \in \cl\W} \|\v - \w\|_1 = 2\gamma(1-\min_i w_i)$ contingent on competitive arbitrage. This metric was studied empirically in~\cite{willetts2026pools} for dynamic-weight pools under a different fee construction.
\end{remark}

\section{Empirical Portfolio Performance}\label{sec:cs}

\begin{table}[b]
\centering
\begin{tabular}{|c|c|c|c|}
\hline
\textbf{Fund} & \textbf{Valuation} & \textbf{Mandate} & \textbf{G3M Dominance Region} \\ \hline\hline
\textbf{VBIAX} & \textbf{NAV} & Economic (Daily TE) & $\emptyset$ \\ \hline
\textbf{VBIAX} & \textbf{NAV} & Economic (Monthly TE) & $\gamma \in [2.73\%,3.90\%]$ \\ \hline\hline
\textbf{EQL} & \textbf{NAV} & Legal & $\emptyset$ \\ \hline
\textbf{EQL} & \textbf{Market Value} & Legal & $\gamma \in [3.56\%,10\%]^\dagger$ \\ \hline
\textbf{EQL} & \textbf{NAV} & Economic & $\gamma \in [3.22\%,7.09\%]$ \\ \hline
\textbf{EQL} & \textbf{Market Value} & Economic & $\gamma \in [3.56\%,10\%]^\dagger$ \\ \hline\hline
\textbf{EDOW} & \textbf{NAV} & Legal & $\emptyset$ \\ \hline
\textbf{EDOW} & \textbf{Market Value} & Legal & $\gamma \in [3.44\%,10\%]^\dagger$ \\ \hline
\textbf{EDOW} & \textbf{NAV} & Economic & $\gamma \in [3.32\%,9.90\%]$ \\ \hline
\textbf{EDOW} & \textbf{Market Value} & Economic & $\gamma \in [3.44\%,10\%]^\dagger$ \\ \hline
\end{tabular}
\caption{Summary table of dominance regions for the G3M over traditional funds. Except where otherwise noted, the TE is computed daily.\\
$\dagger$: End-point of 10\% is the upper range of sampled fees.}
\label{tab:dominance}
\end{table}

Section~\ref{sec:g3m} introduced the theory for the G3M as a band-rebalanced target-weighted fund. In this section, we want to explore the actual practicality of these designs in practice over the range of $\gamma \in [0\%,10\%]$. To do so, we follow the logic of~\cite{willetts2024rebalancing} but consider funds with substantial assets under management as the comparison benchmark. Specifically, we will study a 60-40 fund (VBIAX, Section~\ref{sec:cs-vbiax}), an equal-weighted sector fund (EQL, Section~\ref{sec:cs-eql}), and an equal-weighted DJIA fund (EDOW, Section~\ref{sec:cs-edow}). In contrast to~\cite{willetts2024rebalancing}, we consider the performance of each portfolio construction under two metrics:
\begin{itemize}
\item Fidelity to the mandate as measured by the annualized tracking error (TE) in which we consider two target constructions:
    \begin{itemize}
    \item the underlying index being tracked by the fund;\footnote{VBIAX only includes a single target as it is a mix of two separate indices.}
    \item the daily rebalancing of the underlying \emph{index} positions suggested by the fund's economic description.
    \end{itemize}
The latter target is to measure fidelity against the literal economic mandate while the former is the fund's legal mandate.
\item Compound annual growth rate (CAGR).
\end{itemize}
Importantly, a portfolio is only dominant if it outperforms in both metrics simultaneously (lower TE with higher CAGR). As far as we are aware, no prior work has quantified the TE of an AMM-generated position.
Table~\ref{tab:dominance} summarizes the resulting dominance regions based on G3M fees.

For each of these back-tested positions, we consider the following conventions. The G3M pools are subjected only to arbitrage flow from agents who are able to frictionlessly transact on external markets, thus providing a conservative bound on the fees generated by the order flow (e.g., from uninformed traders). Furthermore, the arbitrageurs transact at most once-per-day at the market close price. Any dividends or other income issued by a position are assumed to be reinvested directly into the asset that paid the income; in this way, total return series for underlying assets and indices can be consistently used throughout.\footnote{All data was obtained via Bloomberg.}

\subsection{VBIAX: 60\% Equity / 40\% Fixed Income}\label{sec:cs-vbiax}

\begin{figure}[t]
\centering
\begin{subfigure}[t]{0.47\textwidth}
\centering
\includegraphics[width=\textwidth]{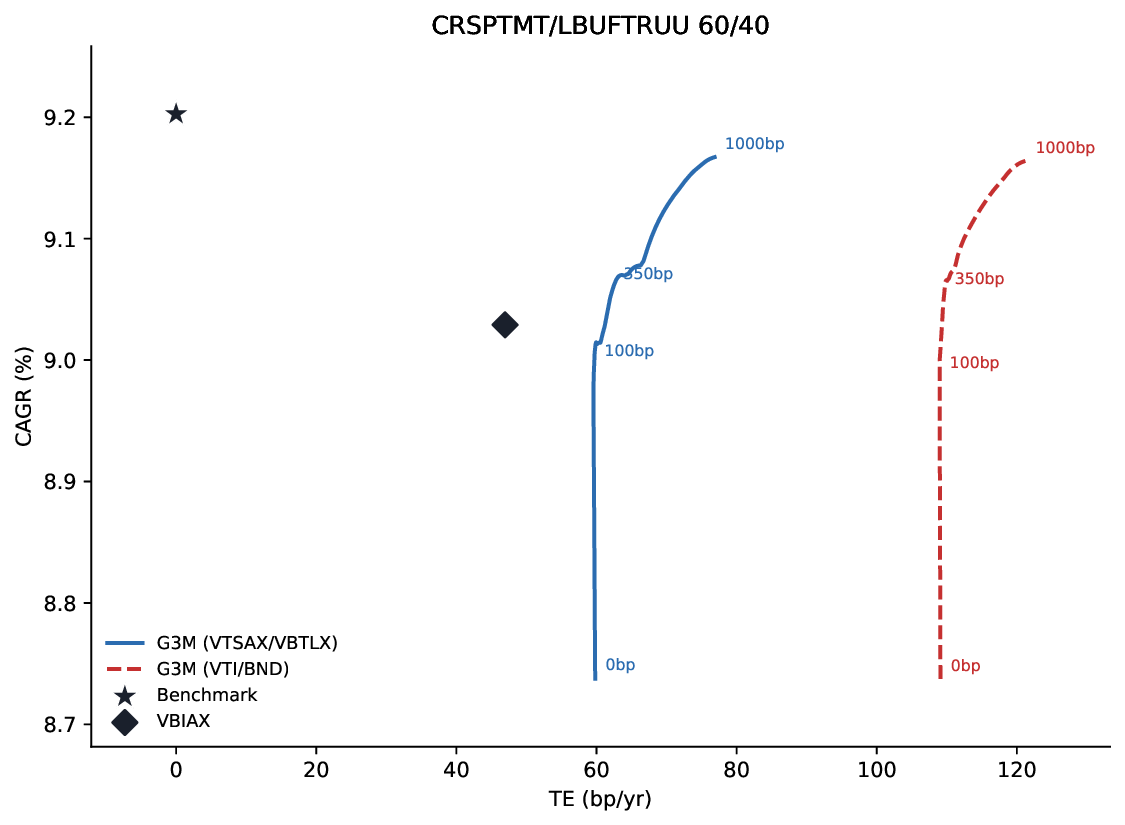}
\caption{Daily TE}
\label{fig:vbiax-daily}
\end{subfigure}
~
\begin{subfigure}[t]{0.47\textwidth}
\centering
\includegraphics[width=\textwidth]{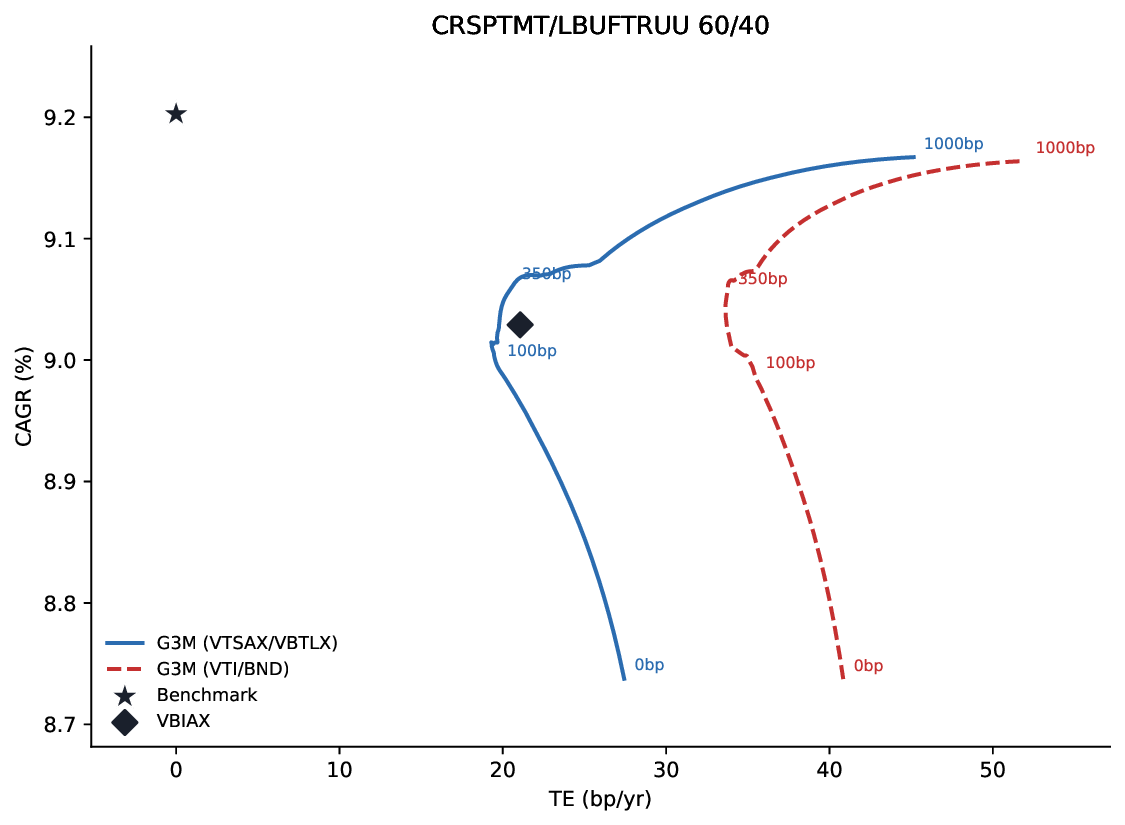}
\caption{Monthly TE}
\label{fig:vbiax-monthly}
\end{subfigure}
\caption{Performance of the G3M against VBIAX on tracking a 60-40 portfolio from January 2, 2014 through June 30, 2026.}
\label{fig:vbiax}
\end{figure}

Consider the VBIAX mutual fund which holds over \$50 billion in assets under management\footnote{\url{https://advisors.vanguard.com/investments/products/vbiax/vanguard-balanced-index-fund-admiral-shares}} which targets holdings that correspond to 60\% equity (Morningstar US Total Market Index, CRSPTMT) and 40\% fixed income (Bloomberg US Aggregate Float Adjusted Index, LBUFTRUU). Though maintaining a consistent economic mandate since its inception, the underlying equity and fixed income indices tracked were modified in 2013; as such, for the purposes of this study we will consider market data from January 2014 through June 2026. As noted in a footnote in the section introduction, herein we consider just the single benchmark constructed from a daily rebalancing of the total return series of CRSPTMT and LBUFTRUU. 

In order to target this same 60-40 portfolio mandate, we consider two different sets of underlying positions for the G3M simulated positions:
\begin{enumerate}
\item the all equity (VTSAX) and all bond (VBTLX) mutual funds; and
\item the all equity (VTI) and all fixed income (BND) ETFs.
\end{enumerate}
Figure~\ref{fig:vbiax} displays the frontier of CAGR vs.\ TE for these simulated G3M positions over the range of fees $\gamma \in [0\%,10\%]$. Notably, we consistently find that the mutual-fund G3M has lower TE than the ETF G3M; we conjecture the increased TE observed by the ETFs is generated by market noise that is not observable in the daily NAV of a mutual fund. While the daily computation of the TE indicates that the G3M never matches the TE of the actual fund VBIAX (Figure~\ref{fig:vbiax-daily}), the monthly computation of TE shows that fees within the range $[2.73\%,3.90\%]$ demonstrate Pareto dominance over VBIAX. This difference in performance based on TE frequency is consistent with the negative autocorrelations observed in the daily active returns, thus leading to an overestimated TE for all funds.

\subsection{EQL: Equal Sector Investment}\label{sec:cs-eql}

\begin{figure}[t]
\centering
\begin{subfigure}[t]{0.47\textwidth}
\centering
\includegraphics[width=\textwidth]{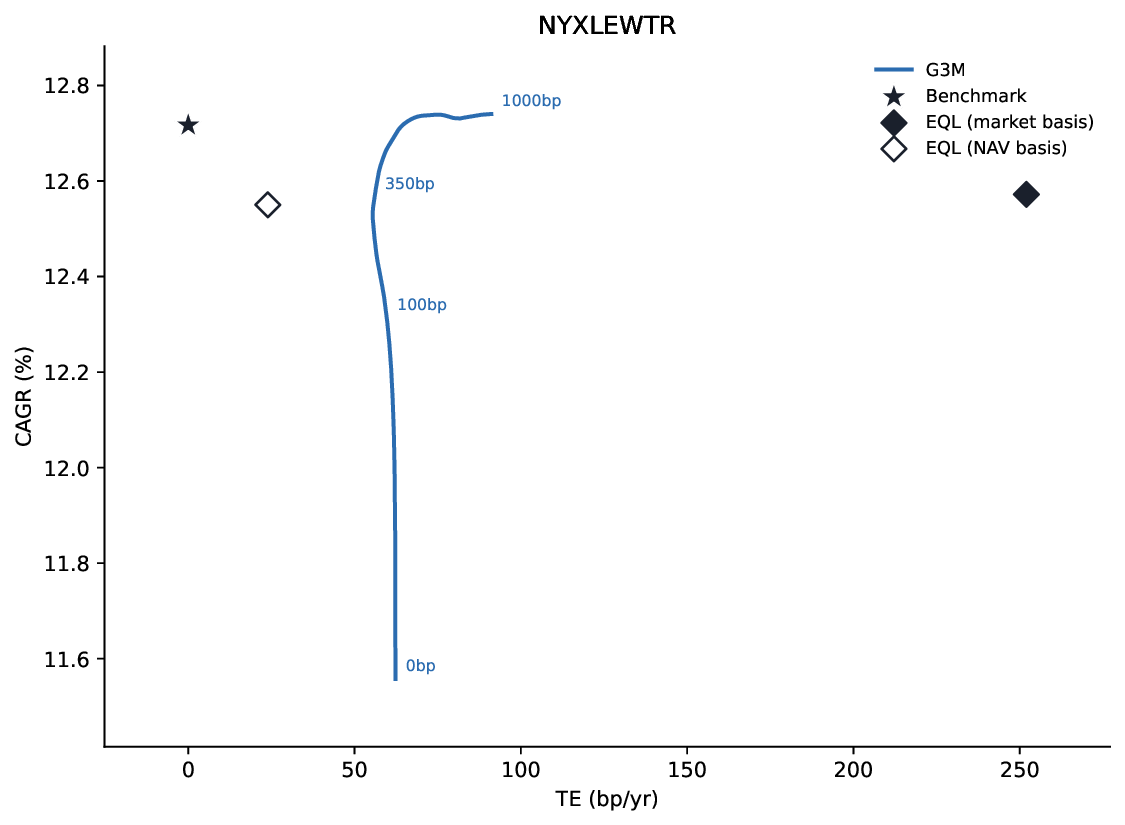}
\caption{TE against the underlying index}
\label{fig:eql-index}
\end{subfigure}
~
\begin{subfigure}[t]{0.47\textwidth}
\centering
\includegraphics[width=\textwidth]{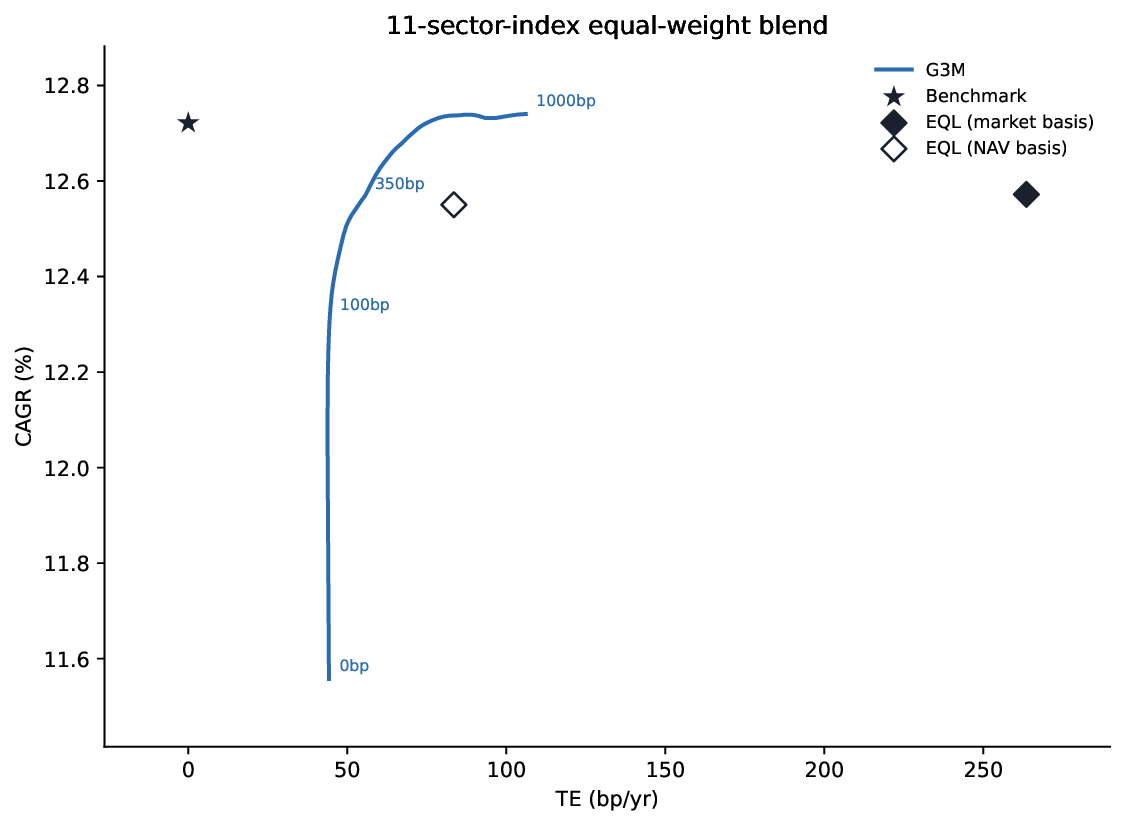}
\caption{TE against the economic mandate}
\label{fig:eql-rebalance}
\end{subfigure}
\caption{Performance of the G3M against EQL on tracking an equal-sector portfolio from June 19, 2018 through May 29, 2026.}
\label{fig:eql}
\end{figure}

Consider the EQL ETF which targets an equal-weighted position in all 11 GICS economic sectors (NYXLEWTR for the total return index). Though maintaining a consistent economic mandate since its inception, the 11th sector was added in mid-2018 while EQL adjusted its target index in mid-2026. As such, for the purposes of this study we will consider market data from June 19, 2018 through May 29, 2026. Notably, we assess the performance of EQL in two ways: (i) based on its market performance and (ii) based on the performance of the NAV. While the market performance is what an investor would realize, the NAV provides a closer approximation to the NAV measurement being used for the G3M position.

In order to assess this same equal-sector portfolio mandate, we consider a G3M holding all 11 Sector SPDR ETFs. Figure~\ref{fig:eql} displays the frontier of CAGR vs.\ TE when considering the two separate benchmarks highlighted in the introduction to this section. First, in Figure~\ref{fig:eql-index}, we observe that the G3M is able to outperform the market performance of EQL in both CAGR and TE for sufficiently high fees, but cannot match the TE of the NAV. However, when the benchmark switches to the daily rebalancing of the underlying index positions (the corresponding sector indices), the G3M can dominate even the NAV position in both CAGR and TE for fees between $[3.22\%,7.09\%]$. While these TE measurements consider daily computations, the monthly figures look comparable and, as such, are not displayed herein.

\subsection{EDOW: Equal-Weighted DJIA}\label{sec:cs-edow}

\begin{figure}[t]
\centering
\begin{subfigure}[t]{0.47\textwidth}
\includegraphics[width=\textwidth]{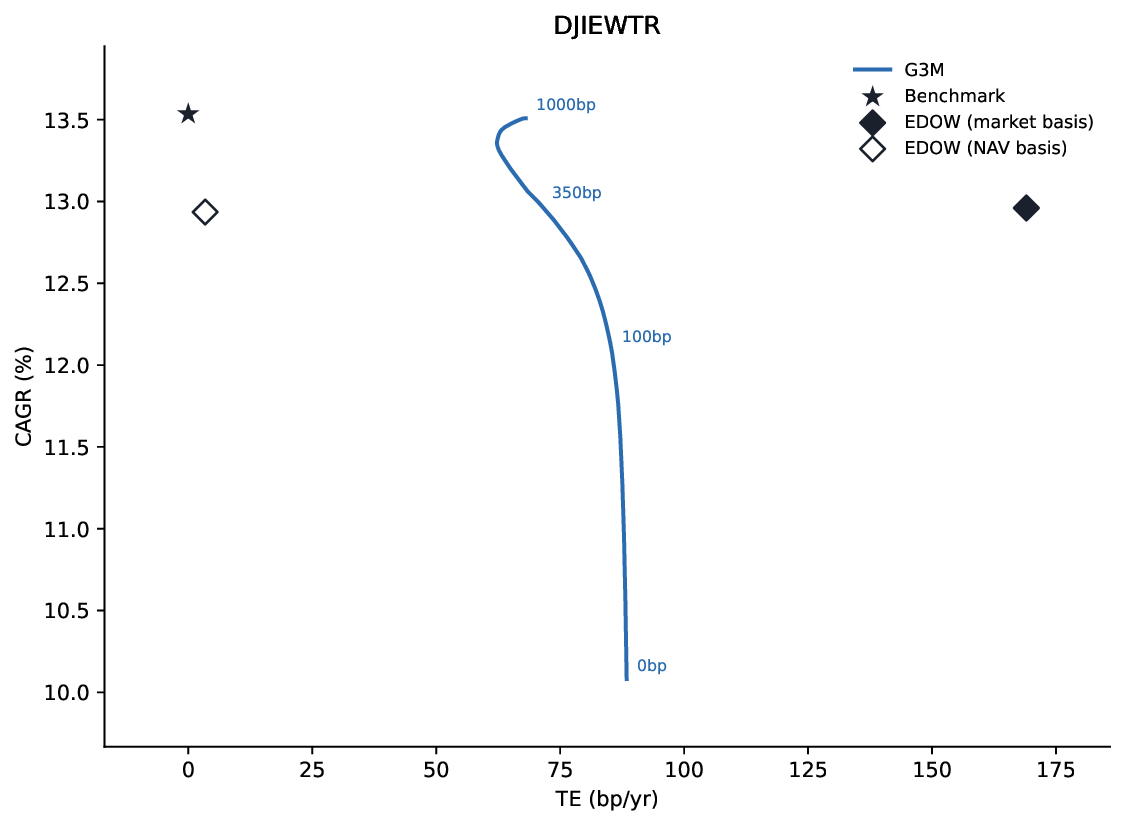}
\caption{TE against the underlying index}
\label{fig:edow-index}
\end{subfigure}
~
\begin{subfigure}[t]{0.47\textwidth}
\includegraphics[width=\textwidth]{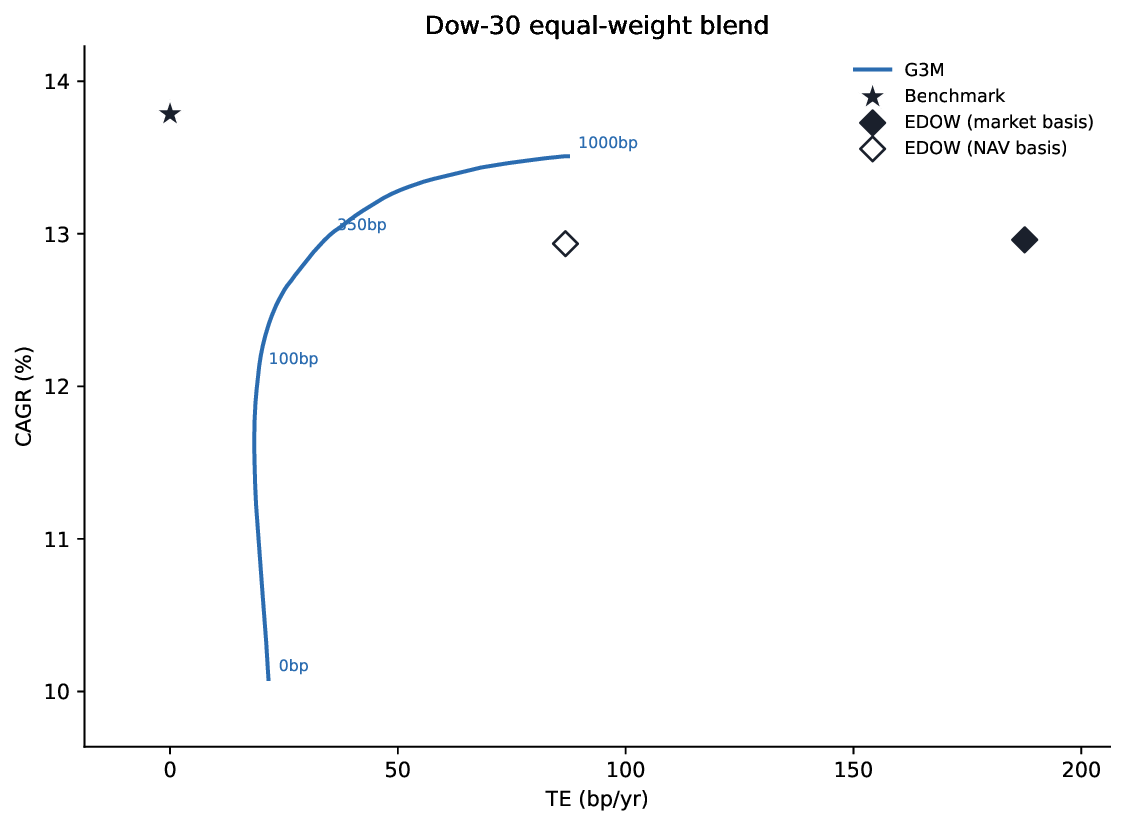}
\caption{TE against the economic mandate}
\label{fig:edow-rebalance}
\end{subfigure}
\caption{Performance of the G3M against EDOW on tracking an equal-weighted DJIA portfolio from November 11, 2024 through June 29, 2026.}
\label{fig:edow}
\end{figure}

Consider the EDOW ETF which targets an equal-weighted position in the 30 stocks that make up the DJIA index (DJIEWTR for the total return index). To eliminate concerns with changing components, we will consider market data from November 11, 2024 to June 29, 2026 during which time the DJIA components were held steady. As with EQL, we assess the performance of EDOW in two ways: (i) based on its market performance and (ii) based on the performance of the NAV.

In order to assess this same equal-weight portfolio mandate, we consider a G3M holding all 30 components of the DJIA. Figure~\ref{fig:edow} displays the frontier of CAGR vs.\ TE when considering the two separate benchmarks highlighted in the introduction to this section. First, in Figure~\ref{fig:edow-index}, we observe that the G3M is able to outperform the market performance of EDOW in both CAGR and TE for sufficiently high fees, but cannot outperform the TE of the NAV. However, when the benchmark switches to the daily rebalancing of the underlying components, the G3M can dominate even the NAV position in both CAGR and TE for fees between $[3.32\%,9.90\%]$. While these TE measurements consider daily computations, the monthly figures look comparable and, as such, are not displayed herein.

\section{Discussion}\label{sec:discussion}

While the theory of Section~\ref{sec:g3m} provides the interpretation of a G3M as a programmatic implementation of a target-weighted portfolio mandate, Section~\ref{sec:cs} demonstrates that this position can be competitive with existing funds. Notably, those existing funds (VBIAX, EQL, and EDOW) have real assets under management which reveal economic demand for the mandate. Therefore, though the prevailing literature asks whether liquidity provision is profitable when viewed as a stand-alone market-making activity, we find that the AMM position can deliver a sufficiently valuable portfolio-management service so as to be economically relevant on its own. From this perspective, the invariant acts as a commitment device, the portfolio is mechanically rebalanced toward its stated mandate without managerial discretion, subject to the no-arbitrage region characterized in Theorem~\ref{thm:tracking-fees}. 

Under this interpretation, the loss-versus-rebalancing (LVR)~\cite{milionis2022automated} takes on a new meaning. The economic transfer from LPs to arbitrageurs is not simply adverse selection, but also the execution costs of maintaining the portfolio mandate and implementing the necessary rebalancing. As in~\cite{willetts2024rebalancing}, the appropriate comparison is not between the AMM and a hypothetical dealer that faces no inventory objective, but between the total performance of the AMM position and the performance of alternative mechanisms for maintaining the same portfolio exposure. As all of the simulated G3M positions in Section~\ref{sec:cs} are generated using only arbitrage flow, the market-making position, measured by fees minus LVR, is non-positive by construction. Nevertheless, the total G3M position can outperform incumbent funds.\footnote{Intriguingly, and subject to future study, within the dominating fee ranges reported in Table~\ref{tab:dominance}, the realized LVR (net of fee revenue) annualizes to a cost of the same order of magnitude as the expense ratio of the corresponding incumbent fund.} The relatively high dominating fee ranges reported in Table~\ref{tab:dominance} (in comparison to fees typically offered by AMMs) are better interpreted as parameters of the portfolio rebalancing band (Corollary~\ref{cor:tracking-fees}) than as competitive exchange spreads.

Furthermore, the empirical results highlight an important distinction between continuous weight fidelity and realized mandate fidelity. Corollary~\ref{cor:tracking-fees} shows that increasing the fee parameter $\gamma$ expands the no-arbitrage region and, therefore, permits larger deviations from the target portfolio weights. However, this does not imply that the realized tracking error increases monotonically with $\gamma$. The $\gamma = 0$ G3M maintains the target weights perfectly, but transfers the full value of each rebalancing opportunity to the arbitrageurs. Positive fees permit temporary mis-weighting while allowing LPs to retain a portion of the proceeds from any rebalancing. This tradeoff produces the non-monotonic frontiers observed in Section~\ref{sec:cs}. 

The comparison between the G3M and traditional index funds is, of course, not immediate. As illustrated by EQL and EDOW (Sections~\ref{sec:cs-eql} and~\ref{sec:cs-edow}), conventional funds are designed to track legally specified indices, including their periodic rebalancing rules. A G3M instead implements the target weights encoded directly in its invariant. The legal index and the continuously maintained target-weight portfolio are, therefore, distinct interpretations of the same portfolio mandate. 

Economically, the G3M resembles an ETF with a custom-basket rule. Although a comparable mandate could in principle be implemented within a traditional fund structure, continuously administering and documenting compliance with such a rule can impose operational costs, which may ultimately be reflected in the fund's expense ratio. With tokenized assets held in a G3M pool, both the mandate and its permissible rebalancing region are encoded directly into the mechanism. Given observable reserves and contemporaneous asset prices, investors and regulators can assess compliance against the bounds in Corollary~\ref{cor:tracking-fees}. Together with the empirical performance found in Section~\ref{sec:cs}, this supports evaluating AMMs not only as decentralized exchanges, but also as verifiable portfolio infrastructure.

\bibliographystyle{plain}
\bibliography{biblio.bib}

\begin{thebibliography}{10}

\bibitem{angeris2020improved}
Guillermo Angeris and Tarun Chitra.
\newblock Improved price oracles: Constant function market makers.
\newblock In {\em Proceedings of the 2nd ACM Conference on Advances in
  Financial Technologies}, pages 80--91, 2020.

\bibitem{angeris2023geometry}
Guillermo Angeris, Tarun Chitra, Theo Diamandis, Alex Evans, and Kshitij
  Kulkarni.
\newblock The geometry of constant function market makers.
\newblock {\em arXiv preprint arXiv:2308.08066}, 2023.

\bibitem{bichuch2026axioms}
Maxim Bichuch and Zachary Feinstein.
\newblock Axioms for automated market makers: A mathematical framework in
  fintech and decentralized finance.
\newblock {\em Operations Research}, 74(3):1187--1202, 2026.

\bibitem{el2023rebalancing}
Rim El~Bernoussi and Michael Rockinger.
\newblock Rebalancing with transaction costs: Theory, simulations, and actual
  data.
\newblock {\em Financial Markets and Portfolio Management}, 37(2):121--160,
  2023.

\bibitem{evans2021liquidity}
Alex Evans.
\newblock Liquidity provider returns in geometric mean markets.
\newblock 2021.

\bibitem{evans2021optimal}
Alex Evans, Guillermo Angeris, and Tarun Chitra.
\newblock Optimal fees for geometric mean market makers.
\newblock In {\em International Conference on Financial Cryptography and Data
  Security}, pages 65--79. Springer, 2021.

\bibitem{martinelli2019non}
Fernando Martinelli and Nikolai Mushegian.
\newblock A non-custodial portfolio manager, liquidity provider, and price
  sensor.
\newblock \url{https://docs.balancer.fi/whitepaper.pdf}, 2019.

\bibitem{milionis2022automated}
Jason Milionis, Ciamac~C Moallemi, Tim Roughgarden, and Anthony~Lee Zhang.
\newblock Automated market making and loss-versus-rebalancing.
\newblock {\em arXiv preprint arXiv:2208.06046}, 2022.

\bibitem{schachermayer2004fundamental}
Walter Schachermayer.
\newblock The fundamental theorem of asset pricing under proportional
  transaction costs in finite discrete time.
\newblock {\em Mathematical Finance: An International Journal of Mathematics,
  Statistics and Financial Economics}, 14(1):19--48, 2004.

\bibitem{willetts2024optimal}
Matthew Willetts and Christian Harrington.
\newblock Optimal rebalancing in dynamic amms.
\newblock {\em arXiv preprint arXiv:2403.18737}, 2024.

\bibitem{willetts2024rebalancing}
Matthew Willetts and Christian Harrington.
\newblock Rebalancing-versus-rebalancing: Improving the fidelity of
  loss-versus-rebalancing.
\newblock {\em arXiv preprint arXiv:2410.23404}, 2024.

\bibitem{willetts2026pools}
Matthew Willetts and Christian Harrington.
\newblock Pools as portfolios: Observed arbitrage efficiency \& lvr analysis of
  dynamic weight amms.
\newblock {\em arXiv preprint arXiv:2602.22069}, 2026.

\end{thebibliography}

\appendix
\section{Proofs}
\subsection{Proof of Lemma~\ref{lemma:path-independence}}\label{proof:lemma:path-independence}
\begin{proof}
First, we refer to~\cite{bichuch2026axioms} for the path independence of the $\gamma = 0$ setting. For the remainder of this proof, we will assume $\gamma \in (0,1)$.

Let $\z := \vec{1} + \vecrho \in \R^N_{++}$ and $\z^{(k)} := \vec{1} + \vecrho^{(k)} \in \R^N_{++}$ for any $k$. In this way, the decomposition satisfies $z_i = \prod_{k = 1}^K z_i^{(k)}$ for every asset $i$. Note also that $\min_i z_i = 1 + \alpha(\vecrho)$ and $\argmin_i z_i^{(k)} = \argmin_i \rho_i^{(k)}$ for every $k$. With this notation,~\eqref{eq:g} becomes
$g_\w^\gamma(\vecrho) = (\min_i z_i)^\gamma \prod_{i = 1}^N z_i^{(1-\gamma)w_i}$.
Immediately, it follows that
\begin{align*}
g_\w^\gamma(\vecrho) &= \left(\min_i z_i\right)^\gamma\prod_{i = 1}^N z_i^{(1-\gamma)w_i} \\
&\geq \left(\prod_{k = 1}^K \min_i z_i^{(k)}\right)^\gamma \prod_{k = 1}^K \left[\prod_{i = 1}^N (z_i^{(k)})^{(1-\gamma)w_i}\right] \\
&= \prod_{k = 1}^K \left[\left(\min_i z_i^{(k)}\right)^\gamma \prod_{i = 1}^N (z_i^{(k)})^{(1-\gamma)w_i} \right].
\end{align*}
Furthermore, equality holds so long as $\min_i z_i = \prod_{k = 1}^K \min_i z_i^{(k)}$; notably this is true if and only if there is a common minimizer across the decompositions $i^* \in \bigcap_{k = 1}^K \argmin_i z_i^{(k)}$ which then is also the minimizer of the original transaction.
\end{proof}

\subsection{Proof of Theorem~\ref{thm:tracking-fees}}\label{proof:thm:tracking-fees}
\begin{proof}
Define $\c := \diag(\s)\x$ to be the component-wise value of the portfolio holdings. For a given arbitrage trade $\Delta\x \in -\x + \R^N_{++}$, we define $\z := \diag(\x)^{-1}[\x+\Delta\x] \in \R^N_{++}$ so that the post-trade reserves in asset $i$ are given by $x_i z_i$. In particular, with this notation, the arbitrageur is equivalently maximizing $\s^\top\x - \c^\top\z$ where the first term is constant with respect to arbitrageur. In this way, we consider the equivalent optimization problem (taking advantage of the monotonicity and positive homogeneity of $g_\w^\gamma$):
\begin{align*}
\min_{\z \in \R^N_{++}} \left\{\c^\top\z \; | \; g_\w^\gamma(\z-\vec{1}) = 1\right\} 
    &= \min_{\z \in \R^N_{++}} \left\{\c^\top\z \; | \; \gamma \min_i \log(z_i) + (1-\gamma)\sum_{i = 1}^N w_i \log(z_i) \geq 0\right\} \\
    &= \min_{\u \in \R^N} \left\{\c^\top\exp(\u) \; | \; \gamma \min_i u_i + (1-\gamma)\sum_{i = 1}^N w_i u_i \geq 0\right\} \\
    &= \min_{\u \in \R^N,t \in \R} \left\{\c^\top \exp(\u) \; | \; \gamma t + (1-\gamma) \sum_i w_i u_i \geq 0, t \leq u_i \forall i\right\}.
\end{align*}
Notably, the final form utilized for the arbitrageur's problem is a convex program for which the KKT conditions are necessary and sufficient for optimality (due to linear-only constraints). Utilizing the Lagrangian of this problem, i.e.,
\begin{align*}
L(\u,t,\lambda,\vecmu) &= \c^\top \exp(\u) - \lambda \left(\gamma t + (1-\gamma) \sum_{i = 1}^N w_i u_i\right) + \sum_i \mu_i (t - u_i) \\
    &= \sum_{i = 1}^N [c_i \exp(u_i) - \underbrace{(\lambda(1-\gamma)w_i + \mu_i)}_{=: \beta_i(\lambda,\mu_i)} u_i] + t \left(\sum_{i = 1}^N \mu_i - \lambda\gamma\right),
\end{align*}
we find the KKT conditions:
    \begin{enumerate}
    \item \emph{Stationarity in $u_i$}: $c_i \exp(u_i) - \beta_i(\lambda,\mu_i) = 0$ for every asset $i$;
    \item \emph{Stationarity in $t$}: $\sum_{i = 1}^N \mu_i - \lambda\gamma = 0$;
    \item \emph{Primal feasibility}: $\gamma t + (1-\gamma) \sum_{i = 1}^N w_i u_i \geq 0,~ t \leq u_i\, \forall i$;
    \item \emph{Dual feasibility}: $\lambda \geq 0,~\vecmu \in \R^N_+$; and
    \item \emph{Complementary slackness}: $\lambda (\gamma t + (1-\gamma) \sum_{i = 1}^N w_i u_i) = 0$ and $\mu_i (t - u_i) = 0$ for every asset $i$.
    \end{enumerate}
Letting $\u^*$ be the outcome of the arbitrageur's optimization, then (following these KKT conditions) the post-arbitrage weights are \[\hat{w}_i = \frac{c_i \exp(u_i^*)}{\c^\top\exp(\u^*)} = \frac{\lambda(1-\gamma)w_i + \mu_i}{\lambda(1-\gamma) + \lambda\gamma} = (1-\gamma)w_i + \frac{\mu_i}{\lambda}.\]
We, further, note that $\lambda = \c^\top\exp(\u^*) > 0$ by these conditions automatically. Therefore, because $\vecmu \in \R^N_+$, it directly follows that $\hat{w}_i \geq (1-\gamma)w_i$ for every asset $i$.

Consider, now, the converse in which the realized weights at $\x$ satisfy the lower weight bounds: $\hat{w}_i \geq (1-\gamma)w_i$ for every asset $i$. We will prove that the arbitrageur's optimal action is to \emph{not} transact in such a setting. For this purpose, consider $\lambda = \sum_{i = 1}^N c_i = \s^\top\x > 0$ and $\vecmu = \lambda(\hat\w - (1-\gamma)\w) \in \R^N_+$. Finally, set $(\u,t) = (\vec{0},0)$. Under these variables, the KKT conditions can all be trivially verified, thus proving the desired result.
\end{proof}

\subsection{Proof of Corollary~\ref{cor:tracking-fees}}\label{proof:cor:tracking-fees}
\begin{proof}
Following the dual variables determined within the proof of Theorem~\ref{thm:tracking-fees}, $\hat\w - (1-\gamma)\w = \vecmu/\lambda \in \gamma\cl\W$. 
Utilizing this simplex structure, the maximal possible value of the weight for any asset $i$ is found when all other assets sit at their minimal no-arbitrage weight, i.e., $(1-\gamma)w_j$ for $j \neq i$. Therefore, $\hat{w}_i \leq 1 - \sum_{j \neq i} (1-\gamma)w_j = \gamma + (1-\gamma) [1 - \sum_{j \neq i} w_j] = \gamma + (1-\gamma) w_i$.
\end{proof}

\end{document}